\documentclass[letterpaper, 10 pt, conference]{ieeeconf}  

\IEEEoverridecommandlockouts                             
\usepackage[english]{babel}

\usepackage[table]{xcolor}
\usepackage{subcaption}
\usepackage{graphics} 
\usepackage{epsfig} 
\usepackage{amsmath, nccmath} 
\usepackage{amssymb}  
\usepackage{amsthm}
\theoremstyle{remark}
\usepackage{amsfonts}

\usepackage{cite}
\usepackage{algorithmic}
\usepackage{graphicx}
\usepackage{textcomp}
\usepackage{xcolor}
\usepackage[nocomma]{optidef}
\def\BibTeX{{\rm B\kern-.05em{\sc i\kern-.025em b}\kern-.08em
    T\kern-.1667em\lower.7ex\hbox{E}\kern-.125emX}}
\usepackage{multirow}
\usepackage{siunitx}
\usepackage{booktabs}
\usepackage{adjustbox}
\usepackage{dsfont}
\usepackage{bm}

\newcommand*{\Scale}[2][4]{\scalebox{#1}{$#2$}}%

\newtheorem{assumption}{Assumption}
\newtheorem{remark}{Remark}
\newtheorem{theorem}{Theorem}
\newtheorem{definition}{Definition}

\title{\LARGE \bf
Interaction-aware Traffic Prediction and Scenario-based Model Predictive Control for Autonomous Vehicles on Highways*
}

\author{Xiaorong Zhang$^{1}$, Sahar Zeinali$^{1}$, and Georg Schildbach$^{1}$
\thanks{*This work was supported by the Deutsche Forschungsgemeinschaft (DFG, German Research Foundation) – project number 460891204.}
\thanks{$^{1}$Institute for electrical engineering in Medicine,University of luebeck, Luebeck, Germany
    {\tt\small \{xiaorong.zhang, sahar.zeinali, georg.schildbach\}@uni-luebeck.de}}%
}

\begin{document}

\maketitle
\thispagestyle{empty}
\pagestyle{empty}

\begin{abstract}

This paper addresses the problem of traffic prediction and control of autonomous vehicles on highways. A modified Interacting Multiple Model Kalman filter algorithm is applied to predict the motion behavior of the traffic participants by considering their interactions. A scenario generation component is used to produce plausible scenarios of the vehicles based on the predicted information. A novel integrated decision-making and control system is proposed by applying a Scenario-based Model Predictive Control approach. The designed controller considers safety, driving comfort, and traffic rules. The recursive feasibility of the controller is guaranteed under the inclusion of the `worst case' as an additional scenario to obtain safe inputs. Finally, the proposed scheme is evaluated using the HighD dataset. Simulation results indicate that the vehicle performs safe maneuvers in different traffic situations under the designed control framework.

\end{abstract}

\section{INTRODUCTION}

\subsection{Motivation}
During the past decades, the design of control systems for autonomous vehicles on highways has been extensively studied. The primary purpose of these systems is to safely control the ego vehicle (EV) by utilizing the predicted motion states of the surrounding target vehicles (TVs) \cite{AVliterature}. The predicted states are usually uncertain, so generating safe, comfortable, energy-efficient, and real-time capable control strategies is challenging.

\subsection{Literature Review}
Physics-based \cite{dynamicmodel}, maneuver-based \cite{intentionestimation}, and interaction-aware motion models \cite{trafficpredictionsurvey} are used for the state prediction of vehicles. Since the interconnections between traffic participants are considered in the interaction-aware models, they are a good choice for describing realistic scenarios. 
Specifically, the mutual influence of the vehicles is usually expressed from a finite set of trajectory clusters or Dynamic Bayesian Networks (DBNs) \cite{interactiveprototype}\cite{interactiveDBN}.
Moreover, a novel interaction-aware traffic model is proposed by Lefkopoulos et al. \cite{lefkopoulos2020}, combining the physics of the vehicles, the intention of the drivers, and a no-collision assumption using an Interacting Multiple Model Kalman filter (IMM-KF), which establishes a new scheme with improved computational efficiency.

Model Predictive Control (MPC) has been widely applied in designing a planner and controller for the EV by considering several constraints, such as the traffic rules, safety, and the comfort of driving \cite{MPC:deterministic, Zhou2022, MPC:laneexchange}. The underlying reason for the prevalence of MPC is its ability to handle explicit constraints in an optimization problem with a moving horizon \cite{MPC}. As a significant variant of Stochastic MPC, Scenario-based MPC (SCMPC) generates corresponding constraints in terms of the possible situational context\cite{SCMPC:lanechangeassistance}. It has been successfully implemented under several highway traffic conditions \cite{SCMPC:lon} \cite{SCMPC:autonomousdriving}, as it is easily compatible with the traffic prediction component and can handle uncertainty using the information, contained in a few scenarios.

Safety is the most critical aspect of controlling the EV. This feature becomes more challenging in emergency scenarios, e.g., unexpected deceleration of the leading vehicle (LV) or a sudden cut-in of a TV. These circumstances are identified as a safety-critical-event (SCE) \cite{SCE:report} \cite{SCE:2017}, where the EV applies immediate braking against the crash, and it may lead to a deceleration of the EV until a standstill. Adaptive Cruise Control (ACC) is a helpful tool for dealing with SCEs, where the EV makes decisions based on the information about the LV. However, a large time headway in this algorithm may lead to over-conservative actions \cite{ACC}. A safety controller is proposed based on ACC in \cite{safety:ACC}, where the EV uses a predefined deceleration profile when the LV performs a maneuver in an SCE manner. Another representative solution to SCEs is the rigorous formalizing mathematical model of Responsibility-Sensitive-Safety (RSS)\cite{RSS_first}. In this model, a safety distance is defined by assuming a 
`worst case' scenario, and the EV responds to an SCE by decelerating at a predefined rate and without applying a full braking force. This approach might be sensitive to the parameter design, and the subtle change of the parameter set might lead to a totally different decision strategy\cite{safety:RSScalibration}.  

\subsection{Contribution}
This paper presents a method for controlling the EV on highways with a comprehensive but simple structure, combining the interaction-aware motion prediction of IMM-KF with the decision-making and control design based on SCMPC. Furthermore, the safety of the EV under the control inputs is guaranteed by theoretically proving the recursive feasibility of the SCMPC under the consideration of the `worst case' scenario. To the best of the authors' knowledge, no work handles all the mentioned problems with an integrated SCMPC architecture and proves the feasibility of the algorithm at the same time. The performance of the proposed approach is evaluated for different HighD dataset scenarios, indicating that EV performs safe and desirable maneuvers by applying the designed control architecture.

\section{CONTROL ARCHITECTURE}
The proposed control architecture is shown in Fig. \ref{fig:control structure} and works as follows. First, the mode states of the TVs are predicted based on the IMM-KF. This information is sent to a scenario generation component to produce all possible maneuvers of the TVs. Then, the most likely scenarios are filtered based on a predefined probability threshold. Finally, an SCMPC-based control system consisting of two control modes is established, which corresponds to 'following the current lane' or 'changing the lane'. In addition to the filtered scenarios, the `worst case' scenario is also included in the design of both controllers. A decision-making module chooses the desired maneuver of the EV based on the cost function value of the two control modes. 
 \begin{figure} [thbp] 
    \centering
    \includegraphics[width=0.5\textwidth]{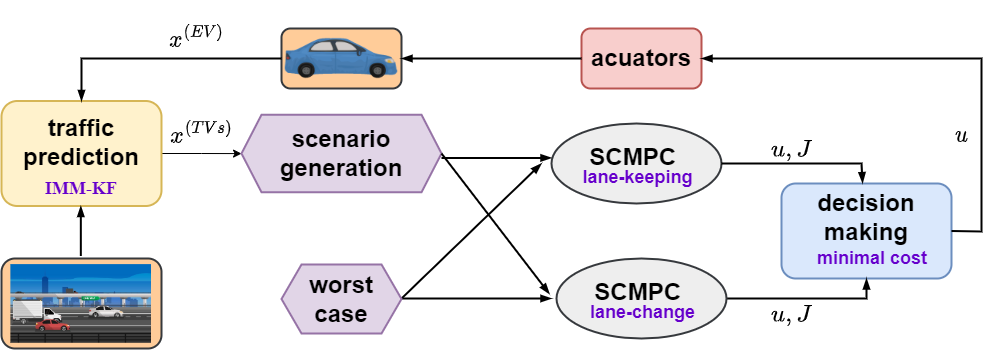}
    \caption{Schematic diagram of the proposed control structure}
    \label{fig:control structure}
\end{figure}

\section{SCENARIO GENERATION}

\subsection{Intention-based Policy Mode}
In this subsection, the longitudinal and lateral policy modes of the vehicles are described. In the longitudinal direction, the 'velocity tracking' (VT) and 'distance keeping' (DK) modes \cite{lefkopoulos2020} are utilized, where the EV tracks a reference velocity in the VT mode and keeps a safe distance from its LV in the DK mode, respectively. In the lateral direction, three modes corresponding to 'lane 1', 'lane 2', and 'lane 3' are applied to represent the target lane of the vehicles, as shown in Fig. \ref{fig:lanes}. For example, if EV is currently at lane 2, it can change to the right lane, lane 1, or the left lane, lane 3, according to the associated lateral mode. Thus, the total number of modes is 6, $M=6$.
\begin{figure}
    \centering
    \framebox{
   \includegraphics[width=0.35\textwidth]{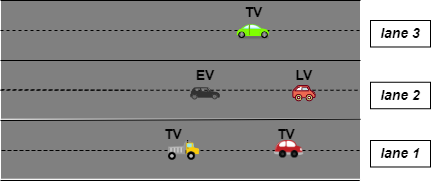}}
    \caption{Illustration of  lanes and vehicles}
    \label{fig:lanes}
\end{figure}

The common state vector $x_k$ in all modes at time $k$ is
\begin{equation}\label{equ:motion state_IMM}
x_k \triangleq
 \begin{bmatrix}
    \underbrace{\begin{matrix} p_{\text{lon},k} & v_{\text{lon},k} & a_{\text{lon},k} \end{matrix}}_{x_{\text{lon},k}} &  \underbrace{\begin{matrix} p_{\text{lat},k} & v_{\text{lat},k} & a_{\text{lat},k} \end{matrix}}_{x_{\text{lat},k}} 
 \end{bmatrix}^\top,
\end{equation}
here $p_{*,k}, v_{*,k}, a_{*,k}$ are respectively the position, velocity, and acceleration in the corresponding direction, $* \in \{\text{lon},\text{lat}\}$. The unknown reference velocity $v_{\text{ref},k}$ (VT mode) or the reference time gap $t_{\text{gap},k}$ (DK mode) is also included in the longitudinal policy mode to be estimated. Therefore, the full state vector $z_k$ in each policy mode at time $k$ is
\begin{equation}\label{equ:motion state_MPC}
z_k \triangleq
 \begin{bmatrix}
    \underbrace{\begin{matrix} x_{\text{lon},k} & r_{\text{ref},k} \end{matrix}}_{z_{\text{lon},k}} &  \underbrace{\begin{matrix} x_{\text{lat},k} \end{matrix}}_{z_{\text{lat},k}} 
 \end{bmatrix}^\top,
\end{equation}
where $r_{\text{ref},k}$ indicates the unknown reference velocity $v_{\text{ref},k}$ or reference time gap $t_{\text{gap},k}$. We use $x_k$ or $z_k$ to denote the common or full states of an arbitrary vehicle, and $x_k^{(\circledast)}$ or $z_k^{(\circledast)}$ to clarify the common or full state of a specific vehicle throughout the paper, where $\circledast  \in$ \{EV, TV, LV\}. To make the policies closer to the driver's real intention, two linear quadratic regulator (LQR)-based feedback controllers are included in longitudinal and lateral policy modes, the associated control gains are defined as follows:
\begin{subequations} \label{equ:controlgain}
\begin{align}
K_{\text{lon}}^{(\diamond)} & =\begin{bmatrix}
    K_{\text{lon},1}^{(\diamond)} & K_{\text{lon},2}^{(\diamond)} & K_{\text{lon},3}^{(\diamond)}
\end{bmatrix}^\top, \\
 K_{\text{lat}} &=\begin{bmatrix}
    K_{\text{lat},1} & K_{\text{lat},2} & K_{\text{lat},3}
\end{bmatrix}^\top.
\end{align}
\end{subequations}
where $\diamond \in \{\text{VT},\text{DK}\}$. Denoting the sampling time as $T$, the discrete form of each policy mode at time $k$ is
\begin{subequations}\label{equ:motion models}
\begin{align}
       \underbrace{\begin{bmatrix}
           z_{\text{lon},k+1} \\
           z_{\text{lat},k+1}
       \end{bmatrix}}_{z_{k+1}}  & = \underbrace{\begin{bmatrix}
            F_{\text{lon},k}^{(\diamond)} & \mathbf{0}_{4\times 3}\\
           \mathbf{0}_{3\times 4} & F_{\text{lat},k}^{(\lambda)}
       \end{bmatrix}}_{F_k}  \underbrace{\begin{bmatrix}
           z_{\text{lon},k} \\
           z_{\text{lat},k}
       \end{bmatrix}}_{z_k} + \underbrace{\begin{bmatrix}
            E_{\text{lon},k}^{(\diamond)} \\
            E_{\text{lat},k}^{(\lambda)}
       \end{bmatrix}}_{E_k} +\omega _k, \\
       y_k &= \bm{I}_{7\times7} z_k + \nu _k,
\end{align}
\end{subequations}
where $y_k \in \mathbb{R}^{7}$ and $\bm{I}_{7\times7}$ are the measurement vector and observation matrix. $\omega_k$ and $\nu _k$ are the process noise and measurement noise, which are assumed to be normally distributed, with the covariance $Q$ and $R$, respectively. The system matrices for longitudinal direction are
\begin{fleqn}
\begin{equation*}\label{equ:statematrices1}
F_{\text{lon},k}^{\text{(VT)}}  = \begin{bmatrix}
    1 & T & \frac{T^2}{2} & 0\\
    0 & 1-\frac{K_{\text{lon},2}^{\text{(VT)}}T^2}{2} & T-\frac{K_{\text{lon},3}^{\text{(VT)}}T^2}{2} & \frac{K_{\text{lon},2}^{\text{(VT)}}T^2}{2}\\
    0 & -K_{\text{lon},2}^{\text{(VT)}}T & 1-K_{\text{lon},3}^{\text{(VT)}}T & K_{\text{lon},2}^{\text{(VT)}}T\\
    0 & 0 & 0 & 1
\end{bmatrix}, 
\end{equation*}
\end{fleqn}
\begin{equation*}\label{equ:statematrices2}
\begin{split}
&F_{\text{lon},k}^ {\text{(DK)}}  =  \\
& \Scale[1]
{\begin{bmatrix}
    1-\frac{K_{\text{lon},1}^{\text{(DK)}}{T}^3}{6} & -\frac{K_{\text{lon},1}^{\text{(DK)}}{T}^2}{2} & -K_{\text{lon},1}^{\text{(DK)}}T & 0\\
    T-\frac{K_{\text{lon},2}^{\text{(DK)}}{T}^3}{6} & 1-\frac{K_{\text{lon},2}^{\text{(DK)}}{T}^2}{2} & -K_{\text{lon},2}^{\text{(DK)}}T & 0\\
    \frac{{T}^2}{2}-\frac{K_{\text{lon},3}^{\text{(DK)}}{T}^3}{6} & T-\frac{K_{\text{lon},3}^{\text{(DK)}}{T}^2}{2} & 1-K_{\text{lon},3}^{\text{(DK)}}T & 0 \\
    -\frac{K_{\text{lon},1}^{\text{(DK)}}v_{\text{lead},k}{T}^3}{6} & -\frac{K_{\text{lon},1}^{\text{(DK)}}v_{\text{lead},k}{T}^2}{2} & -K_{\text{lon},2}^{\text{(DK)}}v_{\text{lead},k}T & 1
\end{bmatrix}^\top}, 
\end{split}
\end{equation*}
and the LQR-related input matrices are 
\begin{fleqn}
\begin{equation*}
E_{\text{lon},k}^{\text{(VT)}} = \mathbf{0}_{4\times1},
\end{equation*}  
\end{fleqn}
\begin{equation*}
\begin{split}
 &\Scale[1]{E_{\text{lon},k}^{\text{(DK)}}} = \\
 &\Scale[1]{\begin{bmatrix}
\frac{K_{\text{lon},1}^{\text{(DK)}}{T}^3}{6}-1 & \frac{K_{\text{lon},2}^{\text{(DK)}}{T}^3}{6}-T & \frac{K_{\text{lon},3}^{\text{(DK)}}{T}^3}{6}-\frac{T^2}{2} & 0\\
\frac{K_{\text{lon},1}^{\text{(DK)}}{T}^2}{2} & \frac{K_{\text{lon},2}^{\text{(DK)}}{T}^2}{2}-1 & \frac{K_{\text{lon},3}^{\text{(DK)}}{T}^2}{2}-T & 0\\
K_{\text{lon},1}^{\text{(DK)}}T & K_{\text{lon},2}^{\text{(DK)}}T & K_{\text{lon},3}^{\text{(DK)}}T-1 & 0 \\
0 & 0 & 0 & 0
\end{bmatrix}\begin{bmatrix}
p_{\text{lon},k}^{(\text{LV})} \\
v_{\text{lon},k}^{(\text{LV})}\\
a_{\text{lon},k}^{(\text{LV})} \\
 0
\end{bmatrix}}.
\end{split}
\end{equation*}
Consider the position $p^{(\lambda)}$ of the center line of the target lane $\lambda$, the system matrix and input matrix for the lateral direction and for the  target lane $\lambda$ ($\lambda= 1,2,3$) are as follows:
\begin{subequations}
\begin{align}
F_{\text{lat},k}^{(\lambda)} &= \begin{bmatrix}
   -\frac{K_{\text{lat},1}{T}^3}{6} & -\frac{K_{\text{lat},2}{T}^3}{6} & -\frac{K_{\text{lat},3}{T}^3}{6}\\
    -\frac{K_{\text{lat},1}{T}^2}{2} & -\frac{K_{\text{lat},2}{T}^2}{2} & -\frac{K_{\text{lat},3}{T}^2}{2}\\
    -K_{\text{lat},1}T & -K_{\text{lat},2}T & -K_{\text{lat},3}T-1 
\end{bmatrix}, \\
E_{\text{lat},k}^{(\lambda)} &= \begin{bmatrix}
    \frac{K_{\text{lat},1}{T}^3}{6}p^{(\lambda)} \\
    \frac{K_{\text{lat},1}{T}^2}{2}p^{(\lambda)}\\  K_{\text{lat},1}Tp^{(\lambda)} 
\end{bmatrix}.
\end{align}
\end{subequations}

\subsection{Interaction-aware Estimation and Prediction}
Since the motion behavior of each vehicle is influenced by its surrounding vehicles, the state estimation and prediction of each vehicle are calculated in descending priority order by considering the interaction between vehicles, as proposed in \cite{lefkopoulos2020}. In particular, the priority criteria are given by
\begin{itemize}
    \item[(i)] If two vehicles are in the same lane, the preceding vehicle has higher priority. 
    \item[(ii)] If two vehicles are in different lanes, the vehicle with the higher longitudinal progress over a specific horizon has higher priority. 
\end{itemize} 

For the sake of clarity, a policy mode corresponds to one model of the IMM-KF in this work. In the IMM-KF \cite{lefkopoulos2020}, we consider the Markov jump linear system \eqref{equ:motion models}, where the transition probability from mode $i$ to mode $j$ is denoted as $\pi^{(i|j)}$, and $\pi^{(i|j)} \in [0,1]$, $i,j\in\{1, 2, ...,6\}$. Since the reference parameter of the VT and DK mode is different, we mix individual common estimates and initialize each mode in the first step as: 
\begin{subequations}
\begin{align}
c^{(i)} & = \sum_{j = 1}^{M} \pi^{(j|i)}\mu_{k-1}^{(j)}, \\
\mu_{k-1}^{(j|i)-} &= \frac{\pi^{(j|i)}\mu_{k-1}^{(j)}}{c^{(i)}}, \\
\Bar{x}_{k-1}^{(i)-} &= \sum_{j = 1}^{M}\mu_{k-1}^{(j|i)-}\hat{x}_{k-1}^{(j)-}, \\
\begin{split}
    \Bar{P}_{k-1}^{(i)-} &= \sum_{j = 1}^{M} \mu_{k-1}^{(j|i)-} \left[ {P}_{k-1}^{(j)-} \right. \\   
    & \left. +\left(\Bar{x}_{k-1}^{(i)-}-\hat{x}_{k-1}^{(j)-}\right)
    \left( \Bar{x}_{k-1}^{(i)-}-\hat{x}_{k-1}^{(j)-} \right)^\top \right],
\end{split}
\end{align}
\end{subequations} 
where $\mu_{k-1}^{(j|i)-}$ is the mixing conditional mode probability, $\hat{x}_{k-1}^{(j)-}$, ${P}_{k-1}^{(j)-}$ are the common state estimation and its associated covariance of the mode $j$, which are part of the full state estimation $\hat{z}_{k-1}^{(j)-}$, and its associated covariance $\mathbb{P}_{k-1}^{(j)-}$. The fused common state estimation and its associated covariance are
$\Bar{x}_{k-1}^{(i)-}$ and $\Bar{P}_{k-1}^{(i)-}$. The corresponding fused full state estimation and its associated covariance are $\Bar{z}_{k-1}^{(i)-}$ and $\Bar{\mathbb{P}}_{k-1}^{(i)-}$. The second step is to predict and update each policy mode as
\begin{subequations}
\begin{align}
       \hat{z}_{k-1}^{(i)+} &= F_{k-1}^{(i)}\Bar{z}_{k-1}^{(i)-}+E_{k-1}^{(i)},\\
       \mathbb{P}_{k-1}^{(i)+} &=  F_{k-1}^{(i)}\Bar{\mathbb{P}}_{k-1}^{(i)-}F_{k-1}^{(i) \top}+Q_{k-1}^{(i)}, \\
       \Tilde{y}_{k}^{(i)} &= y_{k}^{(i)} - \bm{I}_{7\times7}\hat{z}_{k-1}^{(i)+}, \\
       r_{k}^{(i)} &= \bm{I}_{7\times7}\mathbb{P}_{k-1}^{(i)+}\bm{I}_{7\times7}^{\top}+R_{k}^{(i)}, \\
       L_{k}^{(i)} &= \mathbb{P}_{k-1}^{(i)+}\bm{I}_{7\times7}^\top {r_{k}^{(i)}}^{-1}, \\
       \hat{z}_{k}^{(i)-} &= \hat{z}_{k-1}^{(i)+}+L_{k}^{(i)}\Tilde{y}_{k}^{(i)}, \\
       \mathbb{P}_{k}^{(i)-} &= (\bm{I}_{7\times7}-L_{k}^{(i)} \bm{I}_{7\times7})\mathbb{P}_{k-1}^{(i)+} ,
\end{align}
\end{subequations}
using the prior state estimate 
 $\hat{z}_{k-1}^{(i)+} $ and its covariance $\mathbb{P}_{k-1}^{(i)+}$, the innovation residual $ \Tilde{y}_{k}^{(i)}$ and its covariance $ r_{k}^{(i)}$, Kalman gain $L_{k}^{(i)}$, and posterior predicted state estimate 
 $\hat{z}_{k}^{(i)-} $ and covariance $\mathbb{P}_{k}^{(i)-}$. 
 Based on the estimation result of individual mode $ \hat{z}_{k}^{(i)-}$, the state prediction of each policy mode is
\begin{subequations} \label{equ:state prediction}
\begin{align}
   \hat{z}_{t|k} = \phi(t,1) \hat{z}_{k}^{-} + \sum_{\delta  = k+1}^{t}  \phi(t,\delta ) E_{\delta -1},  \\
    \phi(t,\delta ) = \begin{cases}
        (\Pi_{\eta =\delta }^{t-1} F_\eta ^\top)^\top & \text{if} ~ t>\delta  \\
        \bm{I}_{7\times7}, & \text{if} ~ t = \delta 
    \end{cases},  
\end{align}
\end{subequations}
where $t= k+1,...,k+1+N$, and $N$ is the prediction horizon. In order to obtain a collision-free prediction ('no-collision prediction'), a mixed integer quadratic programming (MIQP) problem is formulated to get the modified state estimation $\hat{z}_{k}^{(\text{proj)}-}$, where the safety constraints between the studied vehicle and other vehicles which have higher priority are considered. Note that the state estimation of each policy mode is still $\hat{z}_{k}^{-}$, and only the state prediction is modified in terms of $\hat{z}_{k}^{(\text{proj)}-}$ and \eqref{equ:state prediction}. The state estimation error between 
 $\hat{z}_{k}^{-}$ and $\hat{z}_{k}^{(\text{proj)}-}$, and its covariance are used to augment the innovation residual $\bar{y}_{k}^{{(i)}}$ and its covariance ${{r}_{k}}^{{(i)}}$ as $\breve{y}_{k}^{{(i)}}$ and ${\tilde{r}_{k}}^{{(i)}}$. Then, the policy mode probability is updated based on the augmented matrices 
    \begin{subequations}
        \begin{align}
            \Tilde{L}_{k}^{(i)} &= \frac{ \text{exp}(-\frac{1}{2} \breve{y}_{k}^{{(i)}\top} {\tilde{r}_{k}}^{{(i)-1}} 
            \breve{y}_{k}^{(i)}) }{{\left\lvert 2\pi \tilde{r}_{k}^{(i)}\right\rvert}^{1/2} }, \\
            \Tilde{\mu}_{k}^{(i)} &= \frac{c^{(i)}\Tilde{L}_{k}^{(i)}}{\sum_{j=1}^M c^{(j)}\Tilde{L}_{k}^{(j)}}.
        \end{align}
    \end{subequations}
    The final step is to mix state estimation and its covariance according to the updated probability of the individual mode
\begin{subequations}
\begin{align}
\hat{x}_{k}^{-} &= \sum_{i = 1}^{M}\tilde{\mu}_{k}^{(i)}\hat{x}_{k}^{(i)-}, \\
    P_{k}^{-} &=  \sum_{i = 1}^{M} \Tilde{\mu}_{k}^{(i)}[P_{k}^{(i)-} +(\hat{x}_{k}^{-}-\hat{x}_{k}^{(i)-})(\hat{x}_{k}^{-}-\hat{x}_{k}^{(i)-})^\top].
\end{align}
\end{subequations}
The updated state estimation is also modified in terms of a MIQP problem to guarantee safety over the whole prediction horizon. The readers are referred to Lefkopoulos et al. \cite{lefkopoulos2020} for more details about the mentioned method.

\subsection{Scenario Generation of TVs}
A scenario is defined as a tuple of motion maneuvers for all TVs. Assume that the number of investigated TVs is $V$, then a total of $M^V$ possible scenarios can be generated. $\mu^{(n)}_i$ is the probability of
 TV $n$ with the policy mode $i$, $i \in \{1,2,...,6\}$.
Assuming statistical independence of each vehicle's no-collision prediction over the prediction horizon, then the probability of the scenario $s$ is calculated by
\begin{equation}
     \text{Pr}(s) = \overset{V}{\underset{n=1}{\prod}} {\mu^{(n)}_i}, ~~~ s=1,...,M^V,
\end{equation}
where $\sum_{s=1}^{M^V} \text{Pr}{(s)}= 1$. In order to have high-probability scenarios, scenarios with a probability less than a predefined threshold $\underline{P}$ are not considered. The probability of the remaining scenarios is normalized by
\begin{equation}
    \overline{\text{Pr}}(s) = \frac{\text{Pr}(s)}{1-\sum_{\zeta =1}^{\theta} \text{Pr}(\zeta) }, ~~~ s=1,...,M^V-\theta,
\end{equation}
where $\theta$ is the total number of scenarios with probability less than $\underline{P}$ .

\section{SCENARIO-BASED MODEL PREDICTIVE CONTROL}
Based on the predicted scenarios of the TVs, a feasible trajectory for the EV is calculated by solving a constrained finite-time optimal control problem (CFTOCP) in a moving horizon fashion. The objective of the optimization problem is to follow the planned reference trajectory with minimum effort and with the consideration of safety constraints, traffic rules, and driving comfort. The first computed control input of the CFTOCP is fed to the system at each time step. 

\subsection{Vehicle Model}
The Jerk Model \cite{jerkmodel} is applied to represent the dynamics of the EV:
\begin{equation} \label{equ:vehiclemodel}
    \underbrace{\begin{bmatrix}
       p_{*,k+1}^{\text{(EV)}} \\v_{*,k+1}^{\text{(EV)}} \\ a_{*,k+1}^{\text{(EV)}} 
    \end{bmatrix}}_{x_{*,k+1}^{\text{(EV)}}} = \underbrace{\begin{bmatrix}
        1 & T_\text{p} & \frac{1}{2}{T_\text{p}}^2  \\
        0 & 1 & T_\text{p} \\
        0  & 0 & 1 \\  
    \end{bmatrix}}_A \underbrace{\begin{bmatrix}
       p_{*,k}^{\text{(EV)}} \\v_{*,k}^{\text{(EV)}} \\ a_{*,k}^{\text{(EV)}} 
    \end{bmatrix}}_{x_{*,k}^{\text{(EV)}}} + \underbrace{\begin{bmatrix}
        \frac{1}{6}{T_\text{p}}^3 \\ \frac{1}{2}{T_\text{p}}^2 \\ T_\text{p}
    \end{bmatrix}}_B \underbrace{j_{*,k}^{\text{(EV)}}}_{u_{*,k}^{\text{(EV)}}},
\end{equation}
with the prediction time step ${T_\text{p}}$, the states $x_{*,k}^{\text{(EV)}}$, the position $p_{*,k}^{\text{(EV)}}$, the velocity $v_{*,k}^{\text{(EV)}}$, the acceleration $a_{*,k}^{\text{(EV)}}$, the control inputs $u_{*,k}^{\text{(EV)}}$, and the jerk $j_{*,k}^{\text{(EV)}}$, where $* \in \{\text{lon},\text{lat}\}$. \eqref{equ:vehiclemodel}
can be rewritten as follows:
\begin{equation} \label{equ:vehiclemodel_lonlat}
       x_{k+1}^{\text{(EV)}}  
    = \underbrace{\begin{bmatrix}
        A & 0_{3 \times 3}  \\
        0_{3 \times 3} & A   
    \end{bmatrix}}_{\Bar{A}} x_{k}^{\text{(EV)}}
     + \underbrace{\begin{bmatrix}
        B \\ B 
    \end{bmatrix}}_{\Bar{B}} u_{k}^{\text{(EV)}},
\end{equation}
where $x_k^{\text{(EV)}}$ is defined in \eqref{equ:motion state_IMM}, and $u_k^{\text{(EV)}} =\begin{bmatrix}
    u_{\text{lon},k}^{\text{(EV)}} & u_{\text{lat},k}^{\text{(EV)}}
\end{bmatrix}^\top$. Note that the prediction time step ${T_\text{p}}$ differs from the sampling time step $T$ of the IMM-KF, which usually satisfies $T_\text{p} > T$. 

\subsection{Scenario-based Model Predictive Controller}

Considering lane-keeping and lane-change as possible motion behaviors of the EV, two control modes are proposed leading to different reference trajectories. The first control mode aims to make the EV keep its velocity and stay in the current lane, the second control mode aims to lead the EV to a target lane while maintaining its speed. The decision-making module for choosing the control input resulting from these controllers is integrated into the controller design. The input corresponding to the minimal cost function value is applied to the system. 

In addition to the generated scenarios, a so-called `worst case' scenario is introduced to guarantee the recursive feasibility of CFTOCP. In this scenario, the LV is assumed to be decelerating with its minimum acceleration over the prediction horizon. We introduce two sequences of control inputs $u_{0}^{\text{(EV)}},...,u_{N-1}^{\text{(EV)}}$ and $\breve{u}_{0}^{\text{(EV)}},...,\breve{u}_{N-1}^{\text{(EV)}}$. The first input sequence is calculated by avoiding collision between EV and the LV/TVs under the generated scenarios, which is used to calculate the value of the cost function with its associated states $x_{k}^{\text{(EV)}}$. 
The second input sequence is obtained by considering the safety constraints under the `worst case' scenario, its associated states are $\breve{x}_{k}^{\text{(EV)}}$, and the terminal set of the states is $\breve{\mathbb{X}}^{\text{(EV)}}_f$, which is described in the proof of recursive feasibility. 
The first computed inputs $u_{0}^{\text{(EV)}}$ and $\breve{u}_{0}^{\text{(EV)}}$ must be equal in order to guarantee the recursive feasibility. The CFTOCP is formulated as 
\begin{subequations}
    \begin{align}
    J = & \min_{ u_k^{\text{(EV)}}, x_{k+1}^{\text{(EV)}} } \sum_{k = 0}^{N-1}
    {\left\lVert x_{k+1}^{\text{(EV)}}-x_{\text{ref},{k+1}}^{\text{(EV)}} \right\rVert}_{\Bar{Q}} + {\left\lVert u_k^{\text{(EV)}} \right\rVert}_{\Bar{R}},   \\
        s.t. ~~ &
        \Scale[0.85]{x_{k+1}^{\text{(EV)}} = f(x_k^{\text{(EV)}},u_k^{\text{(EV)}}), ~~k = 0,1,...,N-1,}\\
        &
        \Scale[0.85]{\breve{x}_{k+1}^{\text{(EV)}} = f(\breve{x}_k^{\text{(EV)}},{\breve{u}}_k^{\text{(EV)}}), ~~k = 0,1,...,N-1,}\\     
        &  \Scale[0.85]{x_k^{\text{(EV)}} \in \mathbb{X}^{\text{(EV)}},~~~\breve{x}_k^{\text{(EV)}} \in \breve{\mathbb{X}}^{\text{(EV)}},  ~~  k = 0,1,...,N-1,}   \\  
        & \Scale[0.85]{ u_k^{\text{(EV)}} \in \mathbb{U}^{\text{(EV)}}, ~~~
        {\breve{u}}_k^{\text{(EV)}} \in \breve{\mathbb{U}}^{\text{(EV)}},  ~~ k = 0,1,... ,N-1, }  \\    
        & \Scale[0.85]{ u_{0}^{\text{(EV)}} = \breve{u}_{0}^{\text{(EV)}},} \\
        &\Scale[0.85] {\breve{x}_N^{\text{(EV)}} \in \breve{\mathbb{X}}^{\text{(EV)}}_f} \\
        & \Scale[0.85]{ x_0^{\text{(EV)}} = \breve{x}_0^{\text{(EV)}}=  x^{\text{(EV)}}(0).}
\end{align}
\end{subequations}
Here $x_{\text{ref},{k+1}}^{\text{(EV)}}$ is the reference state of the reference trajectory based on the relevant control mode. Note that we consider the EV only changes one lane in the lane-change mode according to the real traffic situation, so the number of reference trajectories in lane-change mode depends on the current lane of the EV. $ \Bar{Q} \in \mathbb{R}^{6 \times 6} $ and $\Bar{R} \in \mathbb{R}^{2 \times 2}$ are positive definite weighting matrices for tuning. The feasible state sets $\mathbb{X}^{\text{(EV)}}$ and $\breve{\mathbb{X}}^{\text{(EV)}}$, and input set $ \mathbb{U}^{\text{(EV)}}$ and $\breve{\mathbb{U}}^{\text{(EV)}}$ are limited by appropriate constraints, as detailed below. 
\begin{remark} 
   If there is no LV in reality, it is assumed that there is an LV far away from the EV.
\end{remark} 
\begin{remark} \label{rem:deactivatedcases}
During the lane change of the EV, we call the lane-keeping control mode deactivated when keeping the current lane is infeasible.
\end{remark} 

\subsection{Constraints}
 \subsubsection{State and input constraints}
The traffic rules limit the velocity of the vehicles, and the acceleration and jerk are bounded to have a comfortable driving feel. The lateral position of EV is limited by the upper and lower bounds $[l_\text{ub},l_\text{lb}]$ of the lane
\begin{subequations}
\begin{align}
   \noindent  0 <  v_{\text{lon},k}^{\text{(EV)}} , ~~~ & l_{\text{lb}} \le  p_{\text{lat},k}^{\text{(EV)}} \le l_{\text{ub}} \\
   \noindent  \underline{a}_{\text{lon}} 
   \le  a_{\text{lon},k}^{\text{(EV)}} \le \overline{a}_{\text{lon}}, ~~~
   & \underline{a}_{\text{lat}}  \le   a_{\text{lat},k}^{\text{(EV)}} \le \overline{a}_{\text{lat}} , \\
   \noindent \underline{j}_{\text{lon}}  \le  j_{\text{lon},k}^{\text{(EV)}} \le \overline{j}_{\text{lon}}, ~~~ & \underline{j}_{\text{lat}}  \le  j_{\text{lat},k}^{\text{(EV)}} \le \overline{j}_{\text{lat}} ,
\end{align} 
\end{subequations}
where $\underline{\bullet}$ and $\overline{\bullet}$ denote the minimum and maximum values of the associated variables. 

\subsubsection{Safety constraints}
In highway traffic, drivers are required to maintain a safe distance from the preceding vehicles in the same lane, which is translated into the constraint:
 \begin{equation} 
   d_k \geq \underline{d}, 
 \end{equation}
and the safety distance $\underline{d}$ is computed by \cite{distance}
  \begin{equation}\label{equ:safetydistance}
     \underline{d} = \tau v_{\text{lon},k}^{\text{(EV)}} + \bigtriangleup  d,
 \end{equation}
with the design parameters $\tau$ and $\bigtriangleup d$. If the reference point of all vehicles is in their respective center, for example, choose $\bigtriangleup d \geq \frac{l^{(\text{EV})}+l^{(\text{LV})}}{2}$, where $l^{(\text{EV})}$ and $l^{(\text{LV})}$ are the length of EV and LV. During the lane-change period, in addition to keeping a safe distance from the LVs in both the current and target lane, the EV should also maintain a safe distance with the TV behind it in the target lane \cite{SCMPC:lanechangeassistance}. The required safety distance also satisfies \eqref{equ:safetydistance}. The safety constraint under the generated scenario is based on \eqref{equ:safetydistance},  while the safety distance for considering the `worst case' scenario collapses to $\bigtriangleup d$.

\subsection{Recursive Feasibility of the SCMPC}
The proof of the recursive feasibility is able to provide a mathematical guarantee for the feasibility of the designed controller.
\begin{definition} 
 (\textbf{Recursive Feasibility}) The SCMPC in lane-keeping mode is called recursively feasible if a collision with the current LV is always avoidable, while for SCMPC in lane-change mode, it means that no accident occurs between EV and the other vehicles during the lane-change first, and then EV is always safe in the target lane.
\end{definition}

If the safety constraints in the `worst case' scenario are satisfied, it implies that the EV is able to handle all possible traffic circumstances under the SCMPC. The corresponding capability of EV is represented by a parameter, the minimal stopping horizon, which is defined as follows.
\begin{definition}
      (\textbf{Minimal Stopping Horizon}) Given the initial velocity $v_{ {\text{lon}},0}^{\text{(EV)}}$ and the minimal acceleration $\underline{a}_{\text{lon}}^{\text{(EV)}}$ of EV, the minimal stopping horizon $\underline{N} \in \mathbb{N}$ satisfies  
\begin{equation}\label{equ:minimal stopping horizon}
     \underline{N} = \left\lceil  {\frac{v_{ {\text{lon}},0}^{\text{(EV)}}} {\lvert \underline{a}_{\text{lon}}^{\text{(EV)}}  \rvert T_\text{p}}} \right\rceil,
\end{equation}
\end{definition}
\noindent where $\left\lceil \bullet \right\rceil$ is defined as the smallest integer that is not smaller than a real number $\bullet$.

Considering the general traffic situation and rules, we make the following assumption.
\begin{assumption} \label{ass1}
All vehicles only drive forwards, and the EV is only responsible for front collisions.
\end{assumption}
\begin{assumption} \label{ass2}
$u_{k}^{\text{(EV)}}={\begin{bmatrix}
         0 & 0
     \end{bmatrix}}^\top $ is one element of the feasible set $\breve{\mathbb{U}}$.
\end{assumption}

The recursive feasibility of the SCMPC is then proved below. 
\begin{theorem} \label{thm}
If SCMPC is initially feasible, and the prediction horizon $N \geq \underline{N}$, then  the controller is recursively feasible based on Assumptions \ref{ass1}, \ref{ass2}.
\end{theorem}

\begin{proof}
     Let two initial control inputs of the generated normal scenarios and the `worst case' scenario be $\{{{u}_{0|0}^{(\text{EV})}}, {{u}_{1|0}^{(\text{EV})}}, ..., {{u}_{\underline{N}|0}^{(\text{EV})}},...,{{u}_{N|0}^{(\text{EV})}} \}$ and $\{{\breve{u}_{0|0}^{(\text{EV})}}, {\breve{u}_{1|0}^{(\text{EV})}}, ..., {\breve{u}_{\underline{N}|0}^{(\text{EV})}},...,{\breve{u}_{N|0}^{(\text{EV})}} \}$.
     Choose the second control sequences as initially feasible solution $\{{\breve{u}_{0|0}}^{{(\text{EV})}\star}, {\breve{u}_{1|0}}^{{(\text{EV})}\star}, ..., {\breve{u}_{\underline{N}|0}}^{{(\text{EV})}\star},...,{\breve{u}_{N|0}}^{{(\text{EV})}\star} \}$ in the proof, and its related state sequence is $ \{{\breve{x}_{0|0}}^{{(\text{EV})}\star}, {\breve{x}_{1|0}}^{{(\text{EV})}\star}, ..., {\breve{x}_{\underline{N}|0}}^{{(\text{EV})}\star},...,{\breve{x}_{N|0}}^{{(\text{EV})}\star} \}$. Due to the condition \eqref{equ:minimal stopping horizon}, the terminal set $\breve{\mathbb{X}}_{\text{f}}$ satisfies 
    \begin{equation}
        {\breve{x}_{N|0}}^{{(\text{EV})}\star} = \begin{bmatrix} p_{{\text{lon}},\underline{N}|0}^{(\text{EV})} & 0 & 0 & p^{(\lambda)} & 0 & 0
        \end{bmatrix}^{\top},
    \end{equation}
    where the stopping longitudinal position $p_{{\text{lon}},\underline{N}|0}^{(\text{EV})}$ of EV is determined by its initial position  $p_{ {\text{lon}},0}^{\text{(EV)}}$, initial velocity  $v_{ {\text{lon}},0}^{\text{(EV)}}$ and minimal acceleration ${\underline{a}_{\text{lon}}^{\text{(EV)}}}$. Moreover, the terminal lateral position of EV is the position of the center line of the target lane $p^{(\lambda)}$ under the associated control mode.
    According to the system dynamics \eqref{equ:vehiclemodel_lonlat}, we apply ${\breve{u}_{0|0}}^{{(\text{EV})}\star}$ to the system and obtain
    \begin{equation}
        x_{1}^{{(\text{EV})}} = \Bar{A}x_{0}^{{(\text{EV})}} + \Bar{B}{\breve{u}_{0|0}}^{{(\text{EV})}\star} = {\breve{x}_{1|0}}^{{(\text{EV})}\star}. 
    \end{equation}
    Then the following is a feasible solution for the MPC problem initialized at $x_{1}^{{(\text{EV})}}$:
     \begin{equation}
     \begin{split}
         \{{{u}_{1|1}}, {{u}_{2|1}}, ..., {{u}_{\underline{N}|1}},...,{{u}_{N-1|1}},{{u}_{N|1}}\} = 
         \\
         \{{\breve{u}_{1|0}}^{\star}, {\breve{u}_{2|0}}^{\star}, ..., {\breve{u}_{\underline{N}|0}}^{\star},...,{\breve{u}_{N-1|0}}^{\star}, {\begin{bmatrix}
         0 & 0
     \end{bmatrix}}^\top \},
     \end{split}
     \end{equation}
     where ${\begin{bmatrix}
         0 & 0
     \end{bmatrix}}^\top \in \breve{\mathbb{U}}$.
    The corresponding state sequence is 
     \begin{equation}
     \begin{split}
         \{{{x}_{2|1}}, {{x}_{3|1}}, ..., {{x}_{\underline{N}+1|1}},...,{{x}_{N|1}},{{x}_{N+1|1}}\} = 
         \\
         \{{\breve{x}_{2|0}}^{\star}, {\breve{x}_{3|0}}^{\star}, ..., {\breve{x}_{\underline{N}|0}}^{\star},...,{\breve{x}_{N|0}}^{\star}, \bar{A}{\breve{x}_{N|0}}^{\star}+\bar{B}\cdot {\begin{bmatrix}
         0 & 0
     \end{bmatrix}}^\top \},
     \end{split}
     \end{equation}
    where $ \bar{A}{\breve{x}_{N|0}}^{\star}+\bar{B} \cdot {\begin{bmatrix}
         0 & 0
     \end{bmatrix}}^\top = {\breve{x}_{N|0}}^{\star}$.
    Both sequences are feasible for the MPC problem because they satisfy the dynamics and the constraints. 
\end{proof}

\section{SIMULATION AND DISCUSSION}
The presented approach is evaluated with one documentation (ID:$01$) of the HighD Dataset \cite{highDdataset}, which records the motion states of $1047$ vehicles in $900s$ with a sampling time of $0.04s$. We take two specific traffic situations as case studies, including the EV's initial motion states and the TVs' motion states during the simulation time. The designed control system manipulates the motion behavior of the EV in the next time steps. 

\subsection{Simulation Setup}
Fig. \ref{fig:initial_scene} gives the initial traffic scenes of case studies. Table \ref{table:parameters in common} shows the parameters used in the two case studies. The sampling time of the simulation is denoted as $T_\text{s}$. Table \ref{table:parameters individual} shows the initial mode states of vehicles with the proper units, and the length $l$ and width $w$ of vehicles in Cases 1 and 2. 
Note that we only consider EV changes to lane 1 in the lane-change mode of the simulation.
 \begin{figure} [ht] 
    \centering
    \framebox{
    \includegraphics[width=0.45\textwidth]{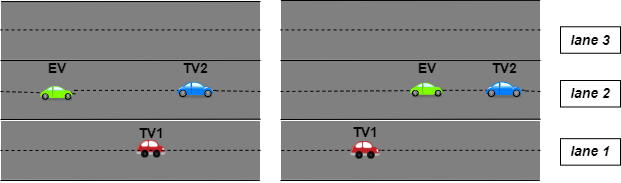}}
    \caption{The initial traffic scenes for Case 1 (left) and Case 2 (right)}
    \label{fig:initial_scene}
\end{figure}

\begin{table} [h!]
   \caption{Design parameters used in the simulation study}
   \rowcolors{1}{gray!50}{white!}
   \label{table:parameters in common}
   \begin{adjustbox}{width=0.48\textwidth}
    \begin{tabular}{|c||c||c||c||c|}
     \hline 
      \multicolumn{2}{|c||}{{ Parameter}} &$T_\text{s}(\text{s})$ & $T(\text{s})$ & $T_\text{p}(\text{s})$   \\
     \hline 
     \multicolumn{2}{|c||}{Value} &  $10.00$ & $0.04$ & $0.40$ \\
      \hline \hline
      \multicolumn{2}{|c||}{{ Parameter}} & $N$ &{$\underline{a}_{\text{lon}}^{(\text{EV})}(\text{m}/\text{s}^2)$} & {$\overline{a}_{\text{lon}}^{(\text{EV})}(\text{m}/\text{s}^2)$}  \\
     \hline 
     \multicolumn{2}{|c||}{Value} & $15$ &  $-4.00$ & $4.00$\\
      \hline \hline
      \multicolumn{2}{|c||}{{ Parameter}} & {$\underline{a}_{\text{lat}}^{(\text{EV})}(\text{m}/\text{s}^2)$}  & {$\overline{a}_{\text{lat}}^{(\text{EV})}(\text{m}/\text{s}^2)$} &{$\underline{j}_{\text{lon}}^{(\text{EV})}(\text{m}/\text{s}^2)$} \\
     \hline 
     \multicolumn{2}{|c||}{Value} & $-4.00$ & $4.00$ &  $-5.00$ \\
     \hline \hline
     \multicolumn{2}{|c||}{{ Parameter}} & {$\overline{j}_{\text{lon}}^{(\text{EV})}(\text{m}/\text{s}^2)$} & {$\underline{j}_{\text{lat}}^{(\text{EV})}(\text{m}/\text{s}^2)$}  & {$\overline{j}_{\text{lat}}^{(\text{EV})}(\text{m}/\text{s}^2)$} \\
     \hline 
     \multicolumn{2}{|c||}{Value} & $5.00$ & $-4.00$ & $4.00$ \\
     \hline \hline
     \multicolumn{2}{|c||}{{ Parameter}} &$\tau(\text{s})$ & $l_\text{ub}(\text{m})$  & $l_\text{lb}(\text{m})$ \\
     \hline 
     \multicolumn{2}{|c||}{Value} &  $0.40$ & $33.80$ & $21.00$ \\
     \hline \hline
     \multicolumn{2}{|c||}{{ Parameter}} & $p_c^{(1),(2),(3)}(\text{m})$ &$K_{\text{lat}}$ & {$\underline{a}_{\text{lon}}^{(\text{LV})}(\text{m}/\text{s}^2)$} \\
    \hline 
     \multicolumn{2}{|c||}{Value} & $\{22.98, 26.88, 31.3\}$ &  $[1.15,3.39,3.58]^\top$ & $-3.00$ \\
     \hline 
       \end{tabular}
  \end{adjustbox}
\end{table}

\begin{table} [h!]
\rowcolors{1}{gray!50}{white!}
\caption{States initial values and vehicle parameters} \label{table:parameters individual}
\begin{adjustbox}{width=0.48\textwidth}
\begin{subtable}{.5\textwidth}
\caption{Case 1}
\centering
   \begin{tabular}{|c||c||c||c|}
     \hline 
     \multicolumn{2}{|c||}{Parameter} & $\{p_{\text{lon},0},v_{\text{lon},0},a_{\text{lon},0},p_{\text{lat},0},v_{\text{lat},0},a_{\text{lat},0}\}$ & $\{l,w\}$  \\
     \hline \hline
      \multicolumn{2}{|c||}{EV} & $\{2.84, 34.8, 0.27, 25.65, -0.09, -0.01\}$ & $\{4.85, 2.02\}$  \\
       \hline \hline
      \multicolumn{2}{|c||}{TV1} & $\{127.87, 32.74, 0.13, 21.52, -0.21, 0.2\}$ & $\{5.96, 2.32\}$  \\
       \hline \hline
      \multicolumn{2}{|c||}{TV2} & $\{141.19, 23.04, 0.05, 25.42, 0.1, -0.03\}$ & $\{14.35, 2.5\}$ \\
      \hline 
   \end{tabular}
\end{subtable}
\end{adjustbox}
\newline
\vspace*{0.2 cm}
\newline
\begin{adjustbox}{width=0.48\textwidth}
\begin{subtable}{.51\textwidth}
\caption{Case 2}
\centering 
   \begin{tabular}{|c||c||c||c|}
   \hline 
     \multicolumn{2}{|c||}{Parameter} & $\{p_{\text{lon},0},v_{\text{lon},0},a_{\text{lon},0},p_{\text{lat},0},v_{\text{lat},0},a_{\text{lat},0}\}$ & $\{l,w\}$  \\
     \hline \hline
      \multicolumn{2}{|c||}{EV} & $\{181.5, 25.07, -0.29, 25.49, 0.09, -0.01\}$ & $\{4.14, 1.92\}$ \\
       \hline \hline
      \multicolumn{2}{|c||}{TV1} & $\{151.54, 32.59, 0.28, 22.21, -0.27, 0.01\}$ & $\{4.75, 2.02\}$ \\
       \hline \hline
      \multicolumn{2}{|c||}{TV2} & $\{201.31, 23.21, 0.17, 25.62, 0.17, -0.02\}$ & $\{9.2, 2.5\}$ \\
      \hline
   \end{tabular}
\end{subtable}
\end{adjustbox}
\end{table}

\subsection{Simulation Results}
Figs. \ref{fig:control result of Case 1} and \ref{fig:control result of Case 2} show the open-loop trajectories of EV over time, the closed-loop trajectory of EV, and TVs' trajectory in Cases 1 and 2. 

\subsubsection{Case 1} The longitudinal and lateral velocity profile of the TVs and the EV throughout the entire sampling time are illustrated in Fig. \ref{fig:case1_vel}. Fig. \ref{fig:cost_case1} displays the cost function value of two controllers, where the associated cost of the lane-keeping or lane-change controller is set as an arbitrary big value, like 5000, when it is deactivated. As Fig. \ref{fig:cost_case1} shows, between $0s$ and $7.6s$, lane-keeping is the desirable maneuver. Between $0s$ and $4s$, the EV keeps its velocity and then decelerates until around $7.6s$ to have a safe distance with TV2. After that, at $7.6s$, the cost of changing the lane becomes lower, prompting the EV to initiate a lane-change maneuver. After approximately $1.5s$, the lane-change maneuver is successfully executed. Subsequently, the EV maintains a constant velocity in its new lane for the remaining time steps, while the lane-change behavior is very costly and becomes impossible at around $11.8s$.
In particular, the increasing values between $9.2s$ and $11.8s$ lead to a visual discontinuity in the lane-change cost-function curve, the largest of which reaches a value of 56000. 

 \begin{figure} [h!] 
    \centering
    \includegraphics[width=0.5\textwidth]{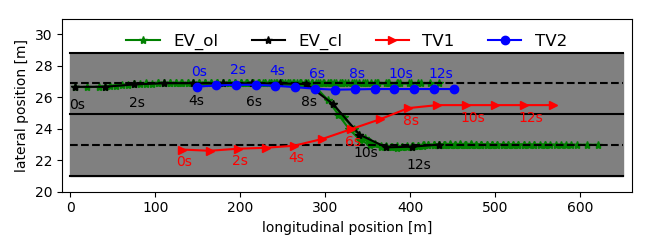}
    \caption{The motion trajectory of vehicles in Case 1}
    \label{fig:control result of Case 1}
\end{figure}
\begin{figure}[h!] 
     \centering
         \includegraphics[width=0.5\textwidth]{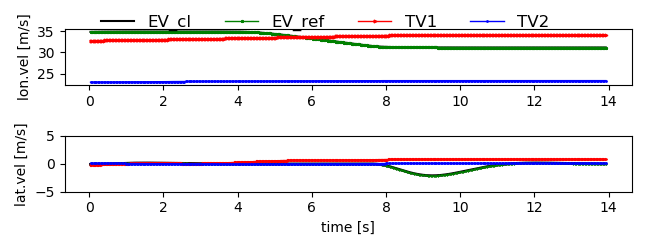}
         \caption{The longitudinal and lateral velocity profile in Case 1}
         \label{fig:case1_vel} 
\end{figure}
\begin{figure}[ht] 
     \centering
         \includegraphics[width=0.5\textwidth]{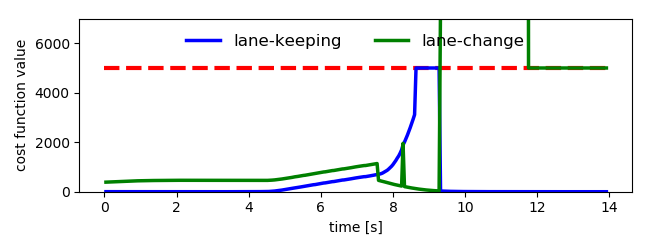}
         \caption{The cost function value of controllers in Case 1}
         \label{fig:cost_case1} 
\end{figure}

The traffic prediction results for Case 1 at $5s$ and $9s$ are shown as two examples in Fig. \ref{fig:case1_5s} and \ref{fig:case1_9s}. There
are three generated scenarios at 5s, which includes TV1 performing a VT maneuver in lane 1, edge color in black, or in lane 2, edge color in blue, and TV2 performing VT maneuvers in lane 1, edge color in red, or in lane 2, edge color in blue. The probabilities of the associated scenarios are $ \{{\mu}^{(1)}_1{\mu}^{(2)}_2,{\mu}^{(1)}_2{\mu}^{(2)}_1,{\mu}^{(1)}_2{\mu}^{(2)}_2\}={\{0.241,0.186,0.573\}}$. 
At $9s$, TV1 might perform the VT maneuver in lane 2, edge color in black, or in lane 3, edge color in magenta, and TV2 might perform a VT maneuver in lane 2, edge color in black. The associated probabilities of the generated scenarios are $\{\mu^{(1)}_2\mu^{(2)}_2,\mu^{(1)}_3\mu^{(2)}_2\} = {\{0.913,0.087\}}$.

\begin{figure}[h!] 
     \centering    
   \includegraphics[width=0.5\textwidth]{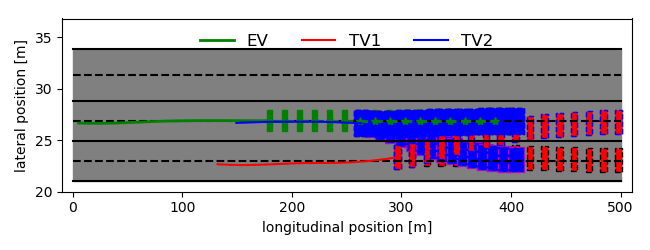}
         \caption{Traffic prediction at $5s$ in Case 1}
    \label{fig:case1_5s}
\end{figure}
\begin{figure}[h!] 
     \centering   
   \includegraphics[width=0.5\textwidth]{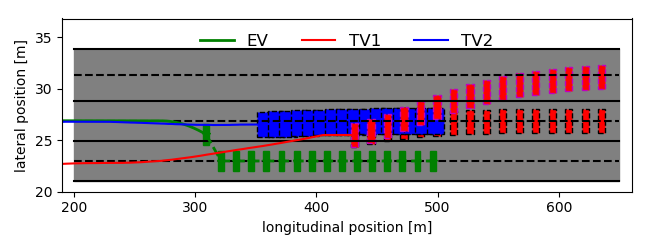}
         \caption{Traffic prediction at $9s$ in Case 1}
    \label{fig:case1_9s}
\end{figure}

\subsubsection{Case 2} The vehicles' motion trajectories, the velocity profile, and the cost function value of two controllers are displayed in Figs. \ref{fig:control result of Case 2}, \ref{fig:case2_vel} and \ref{fig:cost_case2}. Between $0s$ and $2s$, the EV maintains its speed under the control strategy while the cost keeps increasing. The main reason which leads to this situation is that it is impossible to maintain a safe distance from its LV, TV2, by consistently keeping a higher velocity throughout the prediction horizon, so the inconsistency between desired motion states and the calculated states from the controller leads to an increased cost. Since then, the EV starts to reduce its speed to the speed of TV2 until $4s$, which accordingly leads to a decrease in the value of the cost function. 
Note that the lane-change becomes available after about $5s$, before which we set its associated cost function value as 5000. After that, its cost remains higher than lane-keeping. Therefore, EV stays in the current lane in the following time steps. 

  \begin{figure} [h!] 
    \centering
    \includegraphics[width=0.5\textwidth]{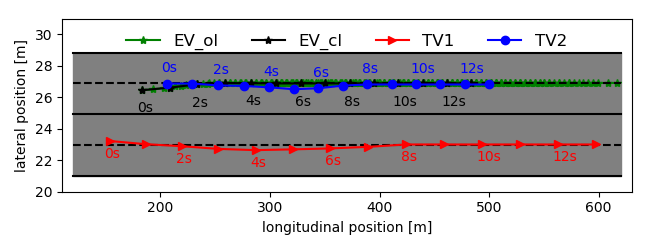}
    \caption{The motion trajectory of vehicles  in Case 2}
    \label{fig:control result of Case 2}
\end{figure}
\begin{figure}[h!]
     \centering
     \includegraphics[width=0.5\textwidth]{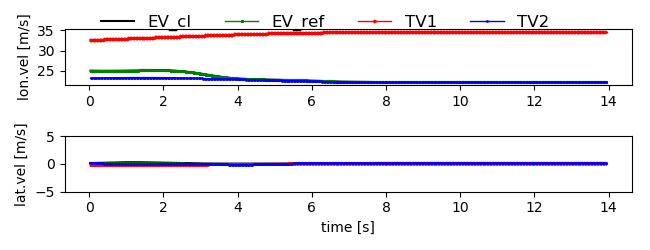}
    \caption{The longitudinal and lateral velocity profile in Case 2}
    \label{fig:case2_vel}
\end{figure}
\begin{figure}[h!]
     \centering
     \includegraphics[width=0.5\textwidth]{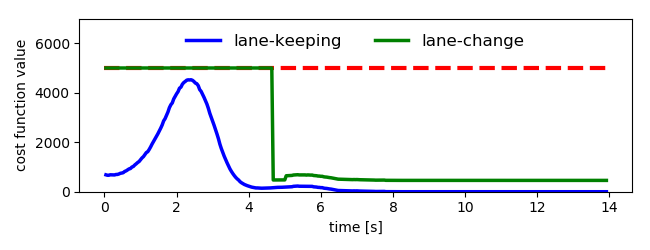}
    \caption{The cost function value of controllers  in Case 2}
    \label{fig:cost_case2}
\end{figure}

The traffic prediction results for Case 2 at $3s$ and $6s$ are shown in Figs. \ref{fig:case2_3s} and \ref{fig:case2_6s}. There are two generated scenarios at $3s$, which includes TV1 performing a VT maneuver in lane 1, edge color in black, or in lane 2, edge color in white, and TV2 performing VT maneuvers in lane 2, edge color in black. The associated probabilities are $ \{\mu^{(1)}_1\mu^{(2)}_2,\mu^{(1)}_2\mu^{(2)}_2\} = {\{0.860,0.140\}}$. At $6s$, there are three considered scenarios, where TV1 might perform a VT maneuver in lane 1, edge color in black, and TV2 might perform a VT maneuver in lane 1, in lane 2, or in lane 3, edge color in black, white, and blue. The corresponding probabilities are $ \{\mu^{(1)}_1\mu^{(2)}_1,\mu^{(1)}_1\mu^{(2)}_2,\mu^{(1)}_1\mu^{(2)}_3\} = {\{0.130,0.727,0.143\}} $.

\begin{figure}[h!] 
     \centering
     \includegraphics[width=0.5\textwidth]{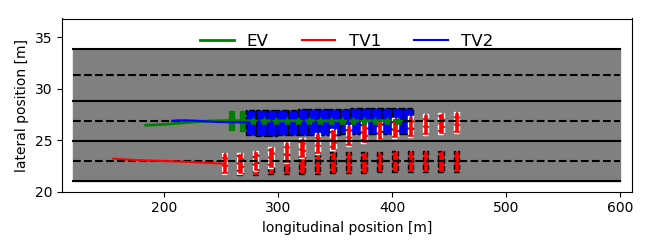}
    \caption{Traffic prediction at $3s$ in Case 2}
    \label{fig:case2_3s}
\end{figure}
\begin{figure}[h!] 
     \centering
        \includegraphics[width=0.5\textwidth]{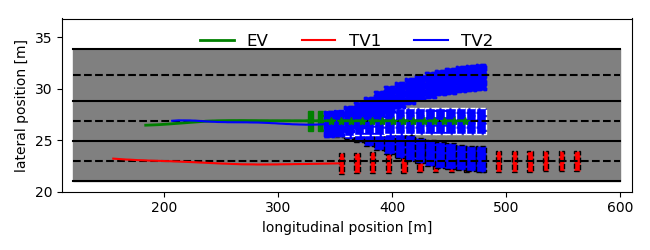}
    \caption{Traffic prediction at $6s$ in Case 2}
    \label{fig:case2_6s}
\end{figure}

Simulation results of Cases 1 and 2 indicate that EV executes safe maneuvers under the designed control structure while considering the interaction with other vehicles.

\section{CONCLUSIONS}
In this paper, the interaction-aware estimation of motion states of the vehicles has been studied using IMM-KF, and the associated state predictions have been combined with the probability to represent the uncertain environment. The generated scenarios, along with the `worst case' scenario, have been applied in formulating the safety constraints of the SCMPC. The control system consists of lane-keeping and lane-change control modes, where the control input with a lower cost function value is implemented into the system. Moreover, the recursive feasibility of the method has been guaranteed based on the no-collision between EV and LV before the minimal stopping horizon of the EV. The proposed algorithm has been validated for two highway scenes chosen from the HighD dataset. The simulation results demonstrate the capability of the proposed control architecture to perform safe maneuvers.






\bibliographystyle{IEEEtran}
\bibliography{main}

\end{document}